\documentclass[a4paper,11pt]{article}

\usepackage{fullpage}
\usepackage{times}
\usepackage{soul}
\usepackage{url}
\usepackage[utf8]{inputenc}
\usepackage[small]{caption}
\usepackage{graphicx}
\usepackage{xcolor}
\usepackage{amsmath}
\usepackage{booktabs}
\usepackage{algorithm}
\usepackage{enumitem}

\usepackage{amsthm}
\usepackage{amssymb}
\usepackage{algpseudocode}
\usepackage{cleveref}

\usepackage{bbm}

\usepackage[textsize=footnotesize,color=green!40]{todonotes}

\newtheorem{theorem}{Theorem}

\newtheorem{lemma}[theorem]{Lemma}

\theoremstyle{definition}

\algnotext{EndFor}
\algnotext{EndIf}
\algnotext{EndWhile}
\algnotext{EndProcedure}

\usepackage{natbib}

\newcommand{\calS}{\mathcal{S}}
\newcommand{\cI}{\mathcal{I}}

\newcommand{\eps}{\varepsilon}
\newcommand{\N}{\mathbb{N}}
\newcommand{\R}{\mathbb{R}}
\newcommand{\pu}{p^{\uparrow}}
\newcommand{\pd}{p^{\downarrow}}
\newcommand{\pl}{p^{\leftarrow}}
\newcommand{\pr}{p^{\rightarrow}}
\newcommand{\ppu}{p'^{\uparrow}}
\newcommand{\ppd}{p'^{\downarrow}}
\newcommand{\ppl}{p'^{\leftarrow}}
\newcommand{\ppr}{p'^{\rightarrow}}
\newcommand{\cB}{\mathcal{B}}
\newcommand{\Vol}{\mathrm{Vol}}

\allowdisplaybreaks

\usepackage{authblk}

\title{\bf On Maximum Bipartite Matching with Separation}

\author[1]{Pasin Manurangsi}
\author[2]{Erel Segal-Halevi}
\author[3]{Warut Suksompong}

\affil[1]{Google Research}
\affil[2]{Ariel University}
\affil[3]{National University of Singapore}

\date{\vspace{-10mm}}

\newcommand{\nodecircle}[1]{
\node [shape=circle,minimum size=5pt,inner sep=0, outer sep=0,fill=black,label=#1]
}

\begin{document}

\maketitle

\begin{abstract}
Maximum bipartite matching is a fundamental algorithmic problem which can be solved in polynomial time.
We consider a natural variant in which there is a separation constraint: the vertices on one side lie on a path or a grid, and two vertices that are close to each other are not allowed to be matched simultaneously.
We show that the problem is hard to approximate even for paths, and provide constant-factor approximation algorithms for both paths and grids.
\end{abstract}

\section{Introduction}
\label{sec:intro}

In the \emph{maximum bipartite matching} problem, we are given a bipartite graph, and our goal is to compute a matching of maximum size in the graph \citep[Sec.~26.3]{CormenLeRi09}.
This is a fundamental algorithmic problem with myriad applications. 
For instance, one may wish to assign meeting rooms in a company to groups of employees for use at a certain hour---an edge between a group and a meeting room indicates that the room is suitable for that group, and it is desirable for as many groups as possible to be assigned. 
Other examples include the allocation of students to time slots for a sports lesson---here, an edge specifies that the corresponding student is available during the associated time slot---or spectators to seats in a theater---an edge between a spectator and a seat means that the spectator finds the seat acceptable.
It is well-known that maximum bipartite matching can be solved in polynomial time \citep{HK73}.

In this paper, we consider a variant of maximum bipartite matching in which one side of the graph admits a spatial or temporal structure and there is a \emph{separation constraint}, meaning that two vertices that are close to each other with respect to the structure should not be matched simultaneously.
Separation constraints are realistic in the aforementioned examples.
Indeed, one may want to avoid assigning neighboring rooms in order to prevent noise disturbance, or placing spectators in adjacent seats due to social distancing requirements.
Similarly, a sports teacher may need to take a break between lessons to clean the facility or refresh herself.
Despite being a natural variation of classic bipartite matching, this problem has not been previously studied to our knowledge.\footnote{Somewhat relatedly, \citet{ElkindSeSu21-Graph,ElkindSeSu21-Land,ElkindSeSu22} investigated the allocation of \emph{divisible} (i.e., non-discrete) resources under separation.}
Can we still compute a maximum matching in polynomial time even in the presence of separation constraints?

\section{Preliminaries}

Let $G = (A,B,E)$ be a bipartite graph, where $A$ and $B$ are the two independent sets of vertices and $E$~is the set of edges; in particular, each edge connects a vertex from~$A$ with a vertex from~$B$.
Let $a = |A|$ and $b = |B|$.
We will assume in \Cref{sec:1D} that the vertices in~$B$ form a one-dimensional structure, and in \Cref{sec:2D} that they form a two-dimensional structure.

We now list some tools from prior work that will be useful for establishing our results.
For any positive integer~$t$, denote by $[t]$ the set $\{1,2,\dots,t\}$.

\paragraph{Maximum $k$-Set Packing.} In the \emph{Maximum $k$-Set Packing} problem, we are given a universe~$U$ along with subsets $S_1, \dots, S_M \subseteq U$, each of size $k$. The goal is to find the maximum number of (pairwise) disjoint subsets.

\begin{theorem}[\cite{Cygan13,furer2014approximating}] \label{thm:set-packing}
For any constants $\eps > 0$ and $k \in \N$ where $k \geq 2$, there is a polynomial-time $(\frac{k + 1}{3} + \eps)$-approximation algorithm for Maximum $k$-Set Packing.
\end{theorem}

We will also use the following APX-hardness of Maximum $3$-Set Packing with bounded degree.\footnote{We remark that \Cref{lem:3sp-hardness} is a restatement of Lemma~4.4 of \cite{GuptaLLM019}, which is stated in terms of partitioning a degree-$4$ graph into vertex-disjoint triangles. The Maximum $3$-Set Packing instance can be created by viewing the vertex set as the universe $U$ and each triangle as a set of three vertices. Since the graph has degree $4$, each element appears in at most ${4 \choose 2} = 6$ subsets.} 
We say that an instance has degree $d$ if each element appears in at most $d$ subsets.
Denote by $\calS$ the collection of the subsets $S_1,\dots,S_M$.

\begin{theorem}[e.g., \cite{GuptaLLM019}]
\label{lem:3sp-hardness}
For some constant $\zeta \in (0, 1/2)$, 
the following problem is NP-hard: 
Given a Maximum $3$-Set Packing instance $(U, \calS)$ with degree $6$, distinguish between\footnote{Note that we can assume without loss of generality that $M = |\calS| \ge |U|/3$ and that $|U|$ is a multiple of~$3$, because we are clearly in the NO case if $M < |U|/3$ or if $|U|$ is not a multiple of $3$.}
\begin{itemize}
\item (YES) There exist $|U|/3$ disjoint subsets;
\item (NO) No $(1 - \zeta)|U|/3$ subsets are disjoint.
\end{itemize}
\end{theorem}

\paragraph{Maximum Independent Sets and $d$-claw-free graphs.} 
Recall that an \emph{independent set} in a graph is a set of vertices such that no two vertices are connected by an edge. 
Finding a maximum-cardinality independent set is a well-known NP-hard problem, which also cannot be approximated to a constant factor unless P $=$ NP \citep{Zuckerman07}. 
However, it admits polynomial-time constant-factor approximation algorithms for a class of graphs called \emph{$d$-claw-free graphs}, which we define next.

In a graph $H = (V_H, E_H)$, a \emph{$d$-claw} consists of a \emph{center} $u \in V_H$ and $d$ \emph{talons} $v_1, \dots, v_d \in V_H$ such that there is an edge between the center and each talon but there is no edge between any pair of talons. 
(That is, $\{u, v_i\} \in E_H$ for all $i \in [d]$ but $\{v_i, v_j\} \notin E_H$ for all $i, j \in [d]$.) 
A graph is \emph{$d$-claw-free} if it does not contain a $d$-claw as an induced subgraph.

\begin{theorem}[\cite{Neuwohner21}] \label{thm:ind}
There exists a constant $\eps > 0$ such that, for any constant $d \in \N$ where $d \geq 4$, there is a polynomial-time $(d/2 - \eps)$-approximation algorithm for Maximum Independent Set in $d$-claw-free graphs.\footnote{A \emph{quasi-polynomial}-time algorithm by \citet{cygan2013sell} finds, for any $\eps>0$, a $(d/3+\eps)$-approximation for Maximum Independent Set in $d$-claw-free graphs.}
\end{theorem}

\section{One Dimension}
\label{sec:1D}

In this section, we consider the one-dimensional setting and assume that the vertices in $B$ lie on a path.
For convenience, we will sometimes associate the vertices in~$B$ with the numbers in the set $[b] = \{1, \dots, b\}$.
There is an integer \emph{separation parameter}~$s\ge 0$; no two vertices in $B$ with fewer than $s$ vertices between them are allowed to be matched simultaneously.
We call the resulting problem \emph{bipartite matching with separation}.
Note that canonical bipartite matching corresponds to the case $s = 0$.

\subsection{Hardness}
Our first result is perhaps surprising: for any $s\ge 1$, deciding whether there exists a matching that involves all vertices in $A$, let alone finding a maximum matching, is already computationally difficult.
In fact, we can even obtain an inapproximability result.

\begin{theorem}
\label{thm:1D-NP}
In the one-dimensional setting, for any fixed $s\ge 1$, there exists a constant $\varepsilon > 0$ such that distinguishing between the following two cases is NP-hard:
\begin{itemize}
\item (YES) There exists a matching of size $a$.
\item (NO) Every matching has size less than $(1-\varepsilon)a$.
\end{itemize}

\end{theorem}

\begin{proof}
The proof is by reduction from Maximum $3$-Set Packing.
Given a Maximum $3$-Set Packing instance $(U = [u], \calS = \{S_1, \dots, S_M\})$ with degree~$6$, where $u$ is a multiple of $3$ and $M \ge u/3$, we construct the bipartite graph $G = (A, B = [b], E)$ as follows (see \Cref{fig:1D-NP}):
\begin{itemize}
\item Let $b = 3(s + 1)M$ and $B=[b]$.
For each $i\in [M]$, we refer to the vertices $3(s + 1)(i - 1) + 1,\dots, 3(s + 1)i$ as \emph{block $B_i$}; so for each set  $S_i\in \calS$, $B$ has a block with $3(s + 1)$ vertices corresponding to $S_i$.
\item Let $A = \{v_1, \dots, v_u\} \cup \{w_1, \dots, w_k\}$, where $k := 2(M - u/3)$.
\item For all $i \in [M]$, and 
for all $u \in S_i$, 
connect the three vertices $3(s + 1)(i - 1) + 1,~~ 3(s + 1)(i - 1) + s + 2,~~ 3(s + 1)(i - 1) + 2s + 3 \in B_i$ to the vertex $v_u \in A$. (See the solid edges in \Cref{fig:1D-NP}.)
\item For each $i \in [M]$, connect the two vertices $3(s + 1)(i - 1) + 2,~~ 3(s + 1)(i - 1) + s + 3 \in B_i$ to all vertices $w_1, \dots, w_k \in A$. (See the dashed edges in \Cref{fig:1D-NP}.)
\end{itemize}
This completes the description of $G$.
Clearly, $G$ can be constructed in polynomial time.

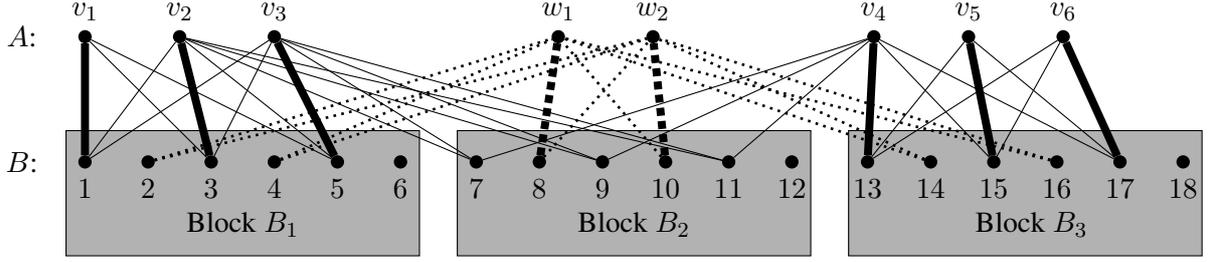
\begin{figure}
\begin{center}
\begin{tikzpicture}[scale=0.83]
\node (A) at (0,2) {$A$:};
\nodecircle{above:$v_1$} (v1) at (1,2) {};
\nodecircle{above:$v_2$} (v2) at (2.5,2) {};
\nodecircle{above:$v_3$} (v3) at (4,2) {};
\nodecircle{above:$v_4$} (v4) at (13.5,2) {};
\nodecircle{above:$v_5$} (v5) at (15,2) {};
\nodecircle{above:$v_6$} (v6) at (16.5,2) {};
\nodecircle{above:$w_1$} (w1) at (8.5,2) {};
\nodecircle{above:$w_2$} (w2) at (10,2) {};
\node (B) at (0, 0) {$B$:};

\draw[fill=black!30] (0.7,-1.5) rectangle ++(5.6,2);
\nodecircle{below:$1$} (b1) at (1,0) {};
\nodecircle{below:$2$} (b2) at (2,0) {};
\nodecircle{below:$3$} (b3) at (3,0) {};
\nodecircle{below:$4$} (b4) at (4,0) {};
\nodecircle{below:$5$} (b5) at (5,0) {};
\nodecircle{below:$6$} (b6) at (6,0) {};
\node[] (T1) at (3.5,-1) {Block $B_1$};

\draw[fill=black!30] (6.9,-1.5) rectangle ++(5.6,2);
\nodecircle{below:$7$} (b7) at (7.2,0) {};
\nodecircle{below:$8$} (b8) at (8.2,0) {};
\nodecircle{below:$9$} (b9) at (9.2,0) {};
\nodecircle{below:$10$} (b10) at (10.2,0) {};
\nodecircle{below:$11$} (b11) at (11.2,0) {};
\nodecircle{below:$12$} (b12) at (12.2,0) {};
\node[] (T2) at (9.7,-1) {Block $B_2$};

\draw[fill=black!30] (13.1,-1.5) rectangle ++(5.6,2);
\nodecircle{below:$13$} (b13) at (13.4,0) {};
\nodecircle{below:$14$} (b14) at (14.4,0) {};
\nodecircle{below:$15$} (b15) at (15.4,0) {};
\nodecircle{below:$16$} (b16) at (16.4,0) {};
\nodecircle{below:$17$} (b17) at (17.4,0) {};
\nodecircle{below:$18$} (b18) at (18.4,0) {};
\node[] (T3) at (16,-1) {Block $B_3$};

\draw[line width=3pt]  (b1) -- (v1);
\draw  (b3) -- (v1);
\draw  (b5) -- (v1);
\draw  (b1) -- (v2);
\draw[line width=3pt]  (b3) -- (v2);
\draw  (b5) -- (v2);
\draw  (b1) -- (v3);
\draw  (b3) -- (v3);
\draw[line width=3pt]  (b5) -- (v3);

\draw  (b7) -- (v2);
\draw  (b9) -- (v2);
\draw  (b11) -- (v2);
\draw  (b7) -- (v3);
\draw  (b9) -- (v3);
\draw  (b11) -- (v3);
\draw  (b7) -- (v4);
\draw  (b9) -- (v4);
\draw  (b11) -- (v4);

\draw[line width=3pt]  (b13) -- (v4);
\draw  (b15) -- (v4);
\draw  (b17) -- (v4);
\draw  (b13) -- (v5);
\draw[line width=3pt]  (b15) -- (v5);
\draw  (b17) -- (v5);
\draw  (b13) -- (v6);
\draw  (b15) -- (v6);
\draw[line width=3pt]  (b17) -- (v6);

\draw[dotted,line width=1pt]  (b2) -- (w1);
\draw[dotted,line width=1pt]  (b2) -- (w2);
\draw[dotted,line width=1pt]  (b4) -- (w1);
\draw[dotted,line width=1pt]  (b4) -- (w2);

\draw[dotted,line width=3pt]  (b8) -- (w1);
\draw[dotted,line width=1pt]  (b8) -- (w2);
\draw[dotted,line width=1pt]  (b10) -- (w1);
\draw[dotted,line width=3pt]  (b10) -- (w2);

\draw[dotted,line width=1pt]  (b14) -- (w1);
\draw[dotted,line width=1pt]  (b14) -- (w2);
\draw[dotted,line width=1pt]  (b16) -- (w1);
\draw[dotted,line width=1pt]  (b16) -- (w2);
\end{tikzpicture}
\end{center}
\caption{
\label{fig:1D-NP}
Graph illustrating the reduction in \Cref{thm:1D-NP}. 
Here, 
$s=1, u=6, M=3, S_1=\{1,2,3\}, S_2=\{2,3,4\}, S_3=\{4,5,6\}, b=18, k=2$.
The thick edges denote the matching induced by the packing $\{S_1, S_3\}$.
}
\end{figure}

\paragraph{(Completeness)} Suppose that there exists $\cI = \{i_1, \dots, i_{u/3}\}$ such that $S_{i_1}, \dots, S_{i_{u/3}}$ are disjoint. 
Since each $S_i$ has size~$3$, this also means that $S_{i_1} \cup \cdots \cup S_{i_{u/3}} = [u]$. 
We construct the matching as follows (see the thick edges in \Cref{fig:1D-NP}).
For each $i \in [M]$:
\begin{itemize}
\item If $i \in \cI$, 
add to the matching three edges connecting the vertices in $A$ corresponding to the elements of $S_i$ to block $B_i$. Specifically, 
let $e_1, e_2, e_3$ denote the elements of $S_i$; add the edges $(v_{e_1}, 3(s + 1)(i - 1) + 1),~~ (v_{e_2}, 3(s + 1)(i - 1) + s + 2),~~ (v_{e_3}, 3(s + 1)(i - 1) + 2s + 3)$.
\item If $i \notin \cI$, add to the matching two edges between block $B_i$ and the vertices $w_j$. Specifically, match $3(s + 1)(i - 1) + 2,~~ 3(s + 1)(i - 1) + s + 3$ to any unmatched vertices among $w_1, \dots, w_k$.
\end{itemize}
It follows from the construction that the separation constraint is respected. 
Note that the second step is valid since there are exactly $M - u/3 = k/2$ indices $i$ outside $\cI$; indeed, this means that all of the vertices $w_1, \dots, w_k$ are matched. 
Furthermore, $S_{i_1} \cup \cdots \cup S_{i_{u/3}} = [u]$ ensures that all of the vertices $v_1, \dots, v_u$ are matched. 
Thus, this is a matching of size $a = |A|$ with separation $s$.

\paragraph{(Soundness)} 
Let $\eps = \zeta/13$, where $\zeta$ is as in \Cref{lem:3sp-hardness}.
Suppose that there exists a matching $W$ with separation $s$ of size at least $(1 - \eps)a$. 
Consider each block~$B_i$; notice that the separation constraint implies that at most three vertices in this block can be matched. 
Let $\cI$ denote the set of indices $i$ such that exactly three vertices in $B_i$ are matched. 
We have
$(1 - \eps)a \leq |W| \leq 3|\cI| + 2(M - |\cI|)$. Plugging in $a = 2M + u/3$ and rearranging, we get
\begin{align*}
|\cI| \geq (1 - \eps)\left(2M + \frac{u}{3}\right) - 2M \geq (1 - \zeta)\cdot\frac{u}{3},
\end{align*}
where the latter inequality follows from our choice of $\zeta$ and from the fact that each element of $U$ appears in at most $6$ sets (implying that $M \leq 6u/3 = 2u$). 
Finally, observe that for each $i \in \cI$, since block~$B_i$ has three vertices matched, these three vertices must be $3(s + 1)(i - 1) + 1,~~ 3(s + 1)(i - 1) + s + 2,~~ 3(s + 1)(i - 1) + 2s + 3$. This implies that the subsets $S_i$ for all $i \in \cI$ are disjoint, and so at least $(1-\zeta)u/3$ subsets are pairwise disjoint. 

\vspace{2mm}

Thus, by \Cref{lem:3sp-hardness}, distinguishing between the YES and NO cases is indeed NP-hard.
\end{proof}

\subsection{Approximation algorithms}

In spite of \Cref{thm:1D-NP}, we show next that it is possible to achieve decent approximation in polynomial time. 
We start with the basic case $s = 1$.
For this case, an easy $2$-approximation algorithm is to find a maximum-cardinality matching (without the separation constraint) and remove either all the odd vertices in $B$ or all the even vertices in $B$ from the matching, whichever there are fewer. 
An improved approximation algorithm is given below.

\begin{theorem}
In the one-dimensional setting, for any fixed $\eps > 0$, there exists a polynomial-time $(4/3 + \eps)$-approximation algorithm for computing a maximum matching of $G$ with separation $s = 1$.
\end{theorem}

\begin{proof}
We create a Maximum $3$-Set Packing instance as follows:
\begin{itemize}
\item Let $U = A \cup \{0, \dots, b\}$ be the universe.
\item For every edge $\{a', b'\} \in E$ where $a'\in A$ and $b'\in B$, create a subset $S_{\{a', b'\}} = \{a', b' - 1, b'\}$.
\end{itemize}
Note that there is a one-to-one correspondence between a matching of $G$ with separation parameter $1$ and a collection of disjoint sets from $\{S_e\}_{e \in E}$. Therefore, we may apply the Maximum $3$-Set Packing algorithm from \Cref{thm:set-packing}, which gives the desired $(4/3 + \eps)$-approximation.
\end{proof}

The next algorithm has a worse approximation ratio, but works for any $s\geq 1$.

\begin{theorem} 
\label{thm:one-dim-better-than-two}
There exists a constant $\eps > 0$ such that, in the one-dimensional setting, for any fixed $s\ge 1$, there exists a polynomial-time $(2 - \eps)$-approximation algorithm for computing a maximum matching of~$G$ with separation  $s$.
\end{theorem}

\begin{proof}
We will reduce to Maximum Independent Set in $4$-claw-free graphs. 
Let $H = (V_H, E_H)$ be the ``conflict graph'' of the edges in $G = (A,B,E)$ defined as follows:
\begin{itemize}
\item Let $V_H = E$.
\item For $\{a',b'\},\{a'',b''\}\in E$, add $\{\{a', b'\}, \{a'', b''\}\}$ to $E_H$ if and only if $a' = a''$ or $|b' - b''| \leq s$.
\end{itemize}
Note that there is a one-to-one correspondence between a matching of $G$ with separation parameter $s$ and an independent set of $H$. 
We claim that $H$ is $4$-claw-free. 
This claim allows us to invoke the algorithm from \Cref{thm:ind}, which immediately yields the desired result.

To see that the claim holds, suppose for the sake of contradiction that there is a $4$-claw~$\Omega$ in $H$, consisting of a center $\{a', b'\}$ and four talons $\{a_1, b_1\}, \{a_2, b_2\}, \{a_3, b_3\}, \{a_4, b_4\}$. 
Since there is an edge between the center and each talon, for each $i \in [4]$, at least one of the following three conditions must hold: (i) $a_i = a'$; (ii) $b_i \in \{b' - s, b' - s + 1, \dots, b'\}$; (iii) $b_i \in \{b' + 1, \dots, b' + s\}$. 
By the pigeonhole principle, at least two talons must satisfy the same condition among the three conditions, so there must be an edge between the vertices corresponding to the two talons in $H$. 
Therefore, $\Omega$ cannot be a $4$-claw, completing the proof of the claim.
\end{proof}

\subsection{Exact algorithm}
Besides approximation, another interesting question is whether we can compute an exact maximum matching efficiently if we impose additional structure on the graph~$G$.
A natural assumption is that each vertex in~$A$ has edges to an \emph{interval} of vertices in~$B$---this reflects scenarios where each student is only available within an interval of consecutive time slots, or each spectator finds a certain block of seats to be acceptable.
We show that this additional structure indeed allows the problem to be solved exactly, by reducing it to a scheduling problem.

\paragraph{Equal-length jobs with release times and deadlines.} 
Suppose we are given $n$ jobs, each job of the same integer length $p \ge 1$. 
The $i$-th job has an integer release time $r_i\geq 0$ and an integer deadline $d_i\geq 1$.  
A schedule for a subset $S \subseteq [n]$ is a list of starting times $s_i$ for each $i \in S$ that respects the constraints, i.e., $s_i \geq r_i$ and $s_i + p \leq d_i$ for each $i\in S$, and $[s_i, s_i + p)\cap [s_j,s_j + p) = \emptyset$ for every pair of distinct $i,j\in S$. 
The \emph{throughput} of a schedule is the number of jobs in the schedule. 

\begin{theorem}[\cite{ChrobakDuJa21}]
 \label{thm:sch}
There exists an algorithm that, given $n$ jobs with release times, deadlines, and identical lengths, computes a schedule that maximizes the throughput in time $O(n^5)$.
\end{theorem}

\begin{theorem}
 \label{thm:one-dim-contiguous-block}
In the one-dimensional setting, for any fixed $s\ge 1$, if each vertex in~$A$ has edges to a (possibly empty) interval of vertices in~$B$, then there exists an algorithm that computes a maximum matching with separation  $s$ in polynomial time.
\end{theorem}

\begin{proof}
We reduce the problem to that of scheduling equal-length jobs with release times and deadlines.
Specifically, we let the length of each job be $p = s + 1$. 
For each vertex $a' \in A$ with edges to a nonempty interval of vertices $x_{a'},x_{a'}+1,\dots,y_{a'}\in B$, we create a job with release time $r_{a'} = x_{a'}$ and deadline $d_{a'} = y_{a'} + p$. 
Note that there is a one-to-one correspondence between a matching of $G$ with separation parameter $s$ and a scheduling\footnote{We may assume without loss of generality that the start times in the schedules are integers, since we can round each start time down to an integer and the schedules remain feasible.} of a subset of the jobs. 
Therefore, applying \Cref{thm:sch} to the constructed jobs yields a maximum matching with separation parameter $s$ in polynomial time.
\end{proof}

\subsection{Extensions}
Next, we briefly discuss some ways in which our positive results can be extended.

\paragraph{Groups of spectators.} Motivated by the application of seating spectators in a theater (see \Cref{sec:intro}), one can consider a setting in which the spectators come in groups of (equal) size $g\geq 1$. 
Each group must be given a contiguous block of $g$ seats, and blocks of different groups must be separated by at least $s$ seats.
\Cref{thm:one-dim-better-than-two} can be adapted to this setting.
In the graph $G$, each vertex in the set $A$ corresponds to a group, and there is an edge between $a'\in A$ and $b'\in B$ if the block of $g$ seats starting at $b'$ is acceptable for group $a'$.
The conflict graph $H$ includes the edge $\{\{a', b'\}, \{a'', b''\}\}$ in $E_H$ if and only if $a' = a''$ 
or $b'' \leq b' \leq b'' + g+ s-1$
or $b' \leq b'' \leq b' + g+ s-1$.
As in the proof of \Cref{thm:one-dim-better-than-two}, this graph $H$ is $4$-claw free, so a polynomial-time $(2 - \eps)$-approximation algorithm exists in this setting as well.
The case where groups may have different sizes remains open.

\paragraph{Bilateral separation constraints.}
\Cref{thm:one-dim-better-than-two} can also be adapted to a setting in which the separation constraint is applied to \emph{both} sides of the graph.
The conflict graph in this case includes the edge $\{\{a', b'\}, \{a'', b''\}\}$ in $E_H$ if and only if $|a' - a''| \leq s$ or $|b' - b''| \leq s$.
By similar arguments to the proof of \Cref{thm:one-dim-better-than-two}, this conflict graph is $5$-claw-free, so a polynomial-time $(2.5 - \eps)$-approximation algorithm exists by \Cref{thm:ind}.

\paragraph{Maximum-weight matching.}
Both \Cref{thm:one-dim-better-than-two,thm:one-dim-contiguous-block} can be extended to a bipartite graph with \emph{weighted} edges, where the goal is to find a matching that maximizes the sum of edge weights, subject to the separation constraint.
\Cref{thm:one-dim-better-than-two} uses the algorithm of \cite{Neuwohner21}, which is already designed for the maximum-weight independent set problem.
For \Cref{thm:one-dim-contiguous-block}, instead of the algorithm by \cite{ChrobakDuJa21}, we can use an algorithm by \citet{baptiste1999polynomial}; this algorithm solves the scheduling problem with equal-length jobs, where the jobs have weights and the goal is to maximize the total weight of the jobs in the schedule, in time $O(n^7)$. 

\section{Two Dimensions}
\label{sec:2D}

In this section, we turn our attention to the two-dimensional case, which is a suitable representation for scenarios such as assigning seats in a theater.
Specifically, we assume that the vertices in~$B$ form a grid, i.e., $B = [\beta_1] \times [\beta_2]$.

\subsection{$\ell_1$ metric} 
Given two points $p = (p_x,p_y), p' = (p_x',p_y') \in \R^2$, their \emph{$\ell_1$ distance}, also known as \emph{Manhattan} or \emph{taxicab distance}, is defined as $\|p - p'\|_1 :=|p_x - p'_x| + |p_y - p'_y|$.
The separation constraint stipulates that two points $p, p' \in B$ should not both be matched if the $\ell_1$ distance between them is at most~$s$.

To facilitate the discussion below, for any point $p \in \R^2$ and real number $r \geq 0$, let 
\[
\cB(p, r) := \{p' \in \R^2 \mid \|p - p'\|_1 \leq r\}.
\]
Furthermore, let $\pu = p + (0, r/2), \pd = p - (0, r/2), \pr = p + (r/2, 0)$, and $\pl = p - (r/2, 0)$. 
We will use the following identity which is not hard to verify (see \Cref{fig:ball-pack}): for any $p \in \R^2$ and real number $r \geq 0$,
\begin{align} \label{eq:ball-pack}
\cB(p, r) = \cB(\pu, r/2) \cup \cB(\pd, r/2) \cup \cB(\pr, r/2) \cup \cB(\pl, r/2).
\end{align}

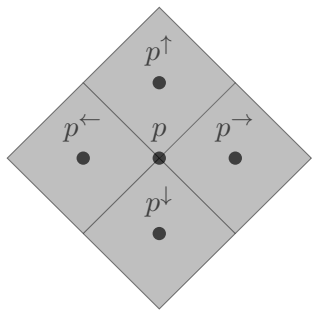
\begin{figure}
\begin{center}
\begin{tikzpicture}	
\nodecircle{$p$} (p) at (0,0) {};
\nodecircle{$\pu$} (pu) at (0,1) {};
\nodecircle{$\pd$} (pd) at (0,-1) {};
\nodecircle{$\pr$} (pr) at (1,0) {};
\nodecircle{$\pl$} (pl) at (-1,0) {};

\draw [opacity=0.5,fill=gray] (0,0) -- (1,1) -- (2,0) -- (1,-1);
\draw [opacity=0.5,fill=gray] (0,0) -- (1,-1) -- (0,-2) -- (-1,-1);
\draw [opacity=0.5,fill=gray] (0,0) -- (-1,-1) -- (-2,0) -- (-1,1);
\draw [opacity=0.5,fill=gray] (0,0) -- (-1,1) -- (0,2) -- (1,1);
\end{tikzpicture}
\end{center}
\caption{
\label{fig:ball-pack}
Covering a ball of radius $r$ in the $\ell_1$ metric with four balls of radius $r/2$.
}
\end{figure}

\begin{theorem}
\label{thm:two-dimensional}
There exists a constant $\eps > 0$ such that, for any fixed $s \geq 1$, in the two-dimensional setting, there exists a polynomial-time $(3 - \eps)$-approximation algorithm for computing a maximum matching of~$G$ with separation parameter $s$ in the $\ell_1$ metric.
\end{theorem}

\begin{proof}
Similarly to the proof of \Cref{thm:one-dim-better-than-two}, we reduce to Maximum Independent Set in $6$-claw-free graphs. Let $H = (V_H, E_H)$ be the ``conflict graph'' of the edges in $G$ defined as follows:
\begin{itemize}
\item Let $V_H = E$.
\item Add $\{\{a', p'\}, \{a'', p''\}\}$ to $E_H$ if and only if $a' = a''$ or $\|p' - p''\|_1 \leq s$.
\end{itemize}
Note that there is a one-to-one correspondence between a matching of $G$ with separation parameter $s$ and an independent set of $H$. 
We claim that $H$ is $6$-claw-free. 
This claim allows us to invoke the algorithm from \Cref{thm:ind}, which gives the desired result.

To see that the claim holds, suppose for the sake of contradiction that there is a $6$-claw $\Omega$ in $H$, consisting of a center $\{a', p'\}$ and six talons $\{a_1, p_1\}, \dots, \{a_6, p_6\}$. 
Since there is an edge between the center to each talon, for each $i \in [6]$, at least one of the following two conditions must hold: $a_i = a'$, or $p_i \in \cB(p', s) = \cB(\ppu, s/2) \cup \cB(\ppd, s/2) \cup \cB(\ppr, s/2) \cup \cB(\ppl, s/2)$ (where the latter equality follows from \Cref{eq:ball-pack}). 
By the pigeonhole principle, there must exist $i\ne j$ such that the two talons $\{a_i, p_i\}, \{a_j, p_j\}$ satisfy $a_i = a_j = a'$ or $p_i, p_j \in \cB(p^*, s/2)$ for some $p^* \in \{\ppu, \ppd, \ppr, \ppl\}$.
In either case, this implies that there is an edge between $\{a_i, p_i\}$ and $\{a_j, p_j\}$ in $H$.
Hence, $\Omega$ cannot be a $6$-claw, completing the proof of the claim.
\end{proof}

\subsection{Other metrics}
\newcommand{\cov}{\operatorname{covnum}}

\Cref{thm:two-dimensional} uses the $\ell_1$ metric only through Equation~\eqref{eq:ball-pack}; it can be extended to other metrics if one can prove a similar ``covering'' result.\footnote{For other metrics, we do not require $s$ to be an integer.}
Formally, for each metric $\mu: \R^2 \times \R^2 \to \R_{\geq 0}$, define its \emph{covering number} $\cov(\mu)$ to be the smallest integer $c$ such that, for every $p \in \R^2$ and $r \in \R_{\geq 0}$, the $r$-radius $\mu$-ball centered at $p$ (i.e., $\cB_\mu(p, r) := \{q \in \R^2 \mid \mu(p, q) \leq r\}$)
can be covered by $c$ sets with $\mu$-diameter~$r$ each. 
Equation~\eqref{eq:ball-pack} shows that $\cov(\ell_1)\leq 4$: an $r$-radius $\ell_1$-ball can be covered by four sets with $\ell_1$-diameter $r$.
\begin{theorem}
\label{thm:two-dimensional-metrics}
There exists a constant $\eps > 0$ such that, for any metric $\mu$ and any fixed $s > 0$, in the two-dimensional setting, there exists a polynomial-time $(\cov(\mu)/2+1 - \eps)$-approximation algorithm for computing a maximum matching of~$G$ with separation $s$ in the metric $\mu$.
\end{theorem}

\begin{proof}
Similarly to the proof of  \Cref{thm:two-dimensional}, we construct the conflict graph $H$ and show that it is $d$-claw-free for $d:=\cov(\mu)+2$. 
Indeed, suppose that any vertex $\{a',p'\}\in H$ is adjacent to $d$ vertices in $H$,
$\{a_1,p_1\},\ldots,\{a_d,p_d\}$.
If two of these vertices $\{a_i,p_i\},\{a_j,p_j\}$ have $a_i = a_j = a'$, then these vertices are adjacent to each other, so $\{a',p'\}$ is not the center of a $d$-claw.
Otherwise, at least $\cov(\mu)+1$ vertices $\{a_i,p_i\}$ have $a_i\neq a'$, which means that  $p_i\in \cB_\mu(p', s)$.
By the pigeonhole principle, 
at least two such points $p_i,p_j$ are contained in the same set of diameter at most $s$ from among the $\cov(\mu)$ sets covering $\cB_\mu(p', s)$. 
Hence, $\{a_i,p_i\}$ and $\{a_j,p_j\}$ are adjacent in $H$, so $\{a',p'\}$ is not the center of a $d$-claw.
\end{proof}

\Cref{thm:two-dimensional} is a special case of \Cref{thm:two-dimensional-metrics}, since $\cov(\ell_1) \leq 4$ by Equation~\eqref{eq:ball-pack}.

To obtain results for other metrics, we need upper bounds on their covering numbers. 
Consider, for example, the $\ell_{\infty}$ metric defined by
$||p-p'||_{\infty} := \max\{|p_x-p'_x|, |p_y-p'_y|\}$.
It is easy to see that $\cov(\ell_{\infty})\leq 4$ (rotate \Cref{fig:ball-pack} by $45^{\circ}$), so \Cref{thm:two-dimensional-metrics} gives a $(3-\eps)$-approximation for this metric as well.
The same holds for an asymmetric variant of $\ell_{\infty}$, defined by $\max\{|p_x-p'_x|/w_x, |p_y-p'_y|/w_y\}$, where $w_x,w_y$ are positive weights.\footnote{To motivate this metric, consider a theater where the distance between adjacent rows is larger than the distance between adjacent columns, so every pair of spectators should be distanced by at least $\lceil w_x\cdot s\rceil$ columns or $\lceil w_y\cdot s\rceil$ rows.}
For this metric, the ball is a rectangle, which can be covered by four similar rectangles scaled by $1/2$ from the original rectangle.

Consider now the Euclidean metric $\ell_2$, defined by
$||p-p'||_{2} := \sqrt{(p_x-p'_x)^2 +  (p_y-p'_y)^2}$.
We have $\cov(\ell_2)\leq 6$,
since the unit disk can be covered by six $60^{\circ}$ sectors, each of which has diameter~$1$. 
This gives a $(4-\eps)$-approximation.

In fact, one can show that every norm-based metric $\mu$ has $\cov(\mu) \leq 25$; we give the proof of this bound below. This implies a $(13.5 - \eps)$-approximation algorithm for any such metric.

\begin{lemma}
For any norm-based metric $\mu$ on $\mathbb{R}^2$, we have $\cov(\mu) \leq 25$.
\end{lemma}

\begin{proof}
By the triangle inequality, every $\mu$-ball of radius $r/2$ has diameter at most $r$. Therefore, it is sufficient to prove that every $\mu$-ball of radius $r$ can be covered by at most $25$ $\mu$-balls of radius $r/2$ each.
Consider the following greedy process of constructing a cover for $\cB_{\mu}(p, r)$: 
\begin{itemize}
\item
Initialize $S := \emptyset$.

\item
While $\cB_{\mu}(p, r) \nsubseteq \bigcup_{p' \in S} \cB_{\mu}(p', r/2)$, add an arbitrary point in $\cB_{\mu}(p, r) \setminus \bigcup_{p' \in S} \cB_{\mu}(p', r/2)$ to $S$.
\end{itemize}
We claim that this process must stop after at most $25$ points have been added to $S$; this immediately yields the desired cover of $\cB_\mu(p, r)$ of size at most $25$. 

Let $v$ denote the volume (i.e., area) of the unit $\mu$-ball.
Due to the construction, 
every new point added to $S$ is at $\mu$-distance larger than $r/2$ from each point already in $S$. Therefore, by the triangle inequality, the balls $\cB_{\mu}(p', r/4)$ for different $p' \in S$ are disjoint. 
As a result, we have\footnote{Here we use the property of norm-based metric spaces in $\R^2$ that $\Vol(\cB_{\mu}(p, r)) = r^2 \cdot v$.}
\begin{align} \label{eq:volume}
\Vol\left(\bigcup_{p' \in S} \cB_{\mu}(p', r/4)\right) = \sum_{p' \in S} \Vol(\cB_{\mu}(p', r/4)) = |S| \cdot v (r/4)^2.
\end{align}
On the other hand, again by the triangle inequality,
we have $\cB_{\mu}(p', r/4) \subseteq \cB_{\mu}(p, 5r/4)$ for all $p' \in S$. 
It follows that
\begin{align*}
\Vol\left(\bigcup_{p' \in S} \cB_{\mu}(p', r/4)\right) \leq \Vol(\cB_{\mu}(p, 5r/4)) = (5r/4)^2 \cdot v = 25r^2 v/16.
\end{align*}
Combining this with \eqref{eq:volume}, we get $|S|\leq 25$, as claimed.
\end{proof}

\begin{theorem}
There exists a constant $\eps > 0$ such that, for any fixed $s > 0$ and any norm-based metric~$\mu$, in the two-dimensional setting, there exists a polynomial-time $(13.5 - \eps)$-approximation algorithm for computing a maximum matching of~$G$ with separation parameter $s$ in the metric~$\mu$.    
\end{theorem}

\section{Future Directions}

Our conflict graph $H$ in the proofs of Theorems~\ref{thm:one-dim-better-than-two}, \ref{thm:two-dimensional}, and \ref{thm:two-dimensional-metrics} is more restricted than a general $d$-claw-free graph, so it might be possible to improve the approximation ratios if we perform more tailored analyses.

Going back to the motivating application of seating spectators in a theater (see \Cref{sec:intro}), one can consider an online setting in which each spectator (or group of spectators) must be seated before the next spectator (or group) arrives.

Matching with separation is a special case of an \emph{independent system of representatives (ISR)} \citep{aharoni2007independent,haxell2011forming}.
Given a graph $H=(V_H,E_H)$ and $n$ sets of vertices $V_1,\ldots,V_n \subseteq V_H$, an ISR is a selection of at most one vertex from each $V_i$ such that the selected vertices form an independent set in $H$.
A matching with separation is in fact an ISR in the conflict graph $H$ induced by the separation constraints (see the proof of \Cref{thm:one-dim-better-than-two}), where the vertex sets $V_1,\ldots,V_n$ are defined by $V_i  := \{\{a',b'\}\in E \mid a'=a_i\}$, that is, the set of edges adjacent to vertex $a_i\in A$.
While the mathematics literature on ISR-s focuses on sufficient conditions for existence, our results provide approximation algorithms for finding an ISR. 
It is interesting whether similar techniques can also be used to derive approximation algorithms for finding ISR-s in more general settings.

\subsection*{Acknowledgments}
This work was partially supported by the Israel Science Foundation under grant number 712/20, by the Singapore Ministry of Education under grant number MOE-T2EP20221-0001, and by an NUS Start-up Grant.
We are grateful to J\'{e}r\^{o}me Lang, Joe Briggs, 
Anton Petrunin, Saul Rodriguez Martin, and Noam D. Elkies for their helpful ideas.

\bibliographystyle{plainnat}
\bibliography{main}

\end{document}